\documentclass[journal]{IEEEtran}
\usepackage{graphicx}
\usepackage{glossaries}
\usepackage{amsthm,amsmath,amssymb,amsfonts}
\usepackage{tikz, circuitikz,tikzscale}
\usepackage{array}
\usepackage{cite}

\newtheorem{proposition}{Proposition}

\newacronym{BS}{BS}{base station}
\newacronym[plural=UEs,
            longplural={user equipment}]{UE}{UE}{user equipment}
\newacronym{UL}{UL}{uplink}
\newacronym{DL}{DL}{downlink}
\newacronym{NOMA}{NOMA}{non-orthogonal multiple access}
\newacronym{CAPA}{CAPA}{continuous aperture array}
\newacronym{OMA}{OMA}{orthogonal multiple access}
\newacronym{LoS}{LoS}{line-of-sight}
\newacronym{SIC}{SIC}{successive interference cancellation}
\newacronym{BER}{BER}{bit error rate}
\newacronym{AWGN}{AWGN}{additive white Gaussian noise}
\newacronym{CF}{CF}{characteristic function}
\newacronym{PD}{PD}{power-domain}
\newacronym{PA}{PA}{power allocation}
\newacronym{MIMO}{MIMO}{multiple-input-multiple-output}
\newacronym{LIS}{LIS}{large intelligent surface}
\newacronym{SIM}{SIM}{stacked intelligent metasurface}
\newacronym{CRIS}{CRIS}{continuous RIS}
\newacronym{HRIS}{HRIS}{holographic reconfigurable intelligent surface}
\newacronym{QAM}{QAM}{quadrature amplitude modulation}
\newacronym{IUI}{IUI}{inter-user interference}
\newacronym{SNR}{SNR}{signal-to-noise ratio}
\newacronym{DRIS}{DRIS}{discrete RIS}

\begin{document}
\title{BER Analysis and Optimization for Continuous RIS-Enabled NOMA}
\author{ Mahmoud AlaaEldin, \textit{Member}, \textit{IEEE}, Amy S. Inwood, \textit{Member}, \textit{IEEE}, \\ Peter J. Smith, \textit{Fellow}, \textit{IEEE}, and Michail Matthaiou, \textit{Fellow}, \textit{IEEE}  \vspace{-0.3in} \\ 

\thanks{This work was supported by the U.K. Engineering and Physical Sciences Research Council (EPSRC) grant (EP/X04047X/2) for TITAN Telecoms Hub. The work of P. J. Smith was supported by the Marsden Fund Council from New Zealand Government funding, managed by Royal Society Te Apārangi. The work of M. Matthaiou was supported by the European Research Council (ERC) under the European Union’s Horizon 2020 research and innovation programme (grant agreement No. 101001331).

M. AlaaEldin, A. S. Inwood, and M. Matthaiou are with the Centre for Wireless Innovation (CWI), Queen’s University Belfast, Belfast BT3 9DT, U.K. (e-mail: \{m.alaaeldin, a.inwood, m.matthaiou\}@qub.ac.uk).

P. J. Smith is with the School of Mathematics and Statistics, Victoria University of Wellington, Wellington, New Zealand (e-mail: peter.smith@vuw.ac.nz).}

}
\maketitle

\begin{abstract}
This letter investigates a novel uplink (UL) system that integrates power-domain non-orthogonal multiple access (PD-NOMA) with a continuous reconfigurable intelligent surface (CRIS). We analyze the effective CRIS-assisted channels under spatially correlated fading to accurately approximate the characteristic function of the cascaded channel. This allows the derivation of an expression for the bit error rate (BER), a key performance metric for UL PD-NOMA. We further utilize the derived BER expressions to introduce a joint optimization framework that minimizes the average BER via UL power allocation and dynamic RIS partitioning among the users. The analytical results are validated by simulations, and show that the proposed optimization scheme eliminates the BER floors that are associated with UL NOMA. The results also confirm the superiority of the optimized CRIS-NOMA scheme over conventional orthogonal multiple access (OMA) and non-optimized UL NOMA schemes.
\end{abstract}


\section{Introduction}

The demand for mobile connectivity continues to grow rapidly, driven by increasingly data-intensive applications, such as high-definition video streaming, augmented reality, and artificial intelligence. However, mobile network resources, such as spectrum and power, are fundamentally limited. To meet this growing demand, it is essential to improve the efficiency with which existing resources are utilized.

One promising technology proposed to address this challenge is \gls{NOMA} \cite{9154358}, which allows multiple \glspl{UE} to simultaneously share the same resource block. In contrast, traditional \gls{OMA} schemes divide the available resources among \glspl{UE}, thereby reducing the portion each \gls{UE} can access.

Continuous aperture surfaces, such as \glspl{HRIS} \cite{chrysologou_outage_2022, chrysologou_when_2023,vo_holographic_2024} and \glspl{CAPA} \cite{liu_capa_2025,hou_transmission_2025}, have significant potential to enhance the performance of \gls{NOMA} systems with minimal power consumption. By manipulating the propagation environment, these surfaces can be employed to strengthen the \gls{PD} separation between \glspl{UE}, improving the overall efficiency of \gls{PD}-\gls{NOMA}.
Continuous aperture surfaces have demonstrated superior performance compared to discrete counterparts \cite{deng_reconfigurable_2022}, offering enhanced beam steering precision and improved utilization of the full physical aperture for energy capture.
Existing work on the symbiosis of continuous aperture surfaces and \gls{NOMA} focuses mainly on \gls{DL} \gls{HRIS}-enabled \gls{NOMA} systems. For this scenario,  \cite{chrysologou_outage_2022} and \cite{chrysologou_when_2023} derived the outage probability for perfect and imperfect \gls{SIC}, respectively, while  \cite{vo_holographic_2024} explored the block error rate. 

To date, there have been no studies on the \gls{UL} or symbol-level performance of continuous aperture surface-aided \gls{NOMA} systems. Therefore, this letter presents the first results on both topics. We build upon the foundational work in \cite{alaaeldin3}, which studied the \gls{BER} of an \gls{UL} \gls{NOMA} system with a \gls{DRIS} under uncorrelated channel conditions. Here, we significantly extend this work to a \gls{NOMA} scenario involving a \gls{CRIS} \cite{inwood_continuous_2025} and the more realistic case of spatially correlated channels. \glspl{CRIS} differ slightly from \glspl{HRIS} and \glspl{DRIS}: a \gls{DRIS} consists of individually controllable spatially separated elements that form a finite physical aperture, an \gls{HRIS} consists of densely packed, electromagnetically coupled elements that form a spatially continuous aperture \cite{chrysologou_when_2023}, whereas a \gls{CRIS} assumes a truly continuous distribution across the aperture \cite{inwood_continuous_2025}. The benefits of studying \glspl{CRIS} are twofold: the derived results serve both as a model for truly continuous surfaces and as an upper bound on the performance of discrete RISs when the element density increases. Considering spatially correlated channels is particularly important for \gls{CRIS} systems, as the infinite element density of a \gls{CRIS} implies that many pairs of points lie within the channel's decorrelation distance. 

The main contributions of this work are as follows:
\begin{itemize}
    \item The received signal of the partitioned \gls{CRIS}  \gls{NOMA} system is expressed in terms of effective channel components, for which we derive the first and second moments.
    
    \item Using these results, we derive the cascaded channel \gls{CF} required to compute the \gls{BER}, a critical metric for \gls{UL} \gls{PD}-\gls{NOMA} due to the error floor.
    
    \item We propose a joint optimization method for \gls{UL} \gls{PA} and \gls{CRIS} partitioning to minimize \gls{BER}.
\end{itemize}

\textit{Notation}: The statistical expectation is denoted $\mathbb{E}[\cdot]$; $\mathrm{Var}[\cdot]$ is the variance; $\mathcal{CN}(\mu,\sigma^2)$ denotes a complex Gaussian distribution with mean $\mu$ and variance $\sigma^2$; $\Re{z}$ and $\Im{z}$ are the real and imaginary components of a complex number $z$, respectively; $(\cdot)^*$ represents the complex conjugate operator; $\angle{z}$ is the angle of a complex number $z$; ${}_2F_1(\cdot)$ is the Gaussian hypergeometric function.

\section{System Model}  \label{sysmod}
We consider the \gls{UL} system in Fig.~\ref{fig:sys_model}, where $K$ single-antenna \glspl{UE}, $U_1$ to $U_K$, communicate with a single-antenna \gls{BS} via a passive, rectangular \gls{CRIS} of width $W$ and height $H$, upon which every point can effect the desired phase shift.\footnote{While physically realizing a perfect \gls{CRIS} is challenging, hardware designs have recently been proposed to closely implement a continuous aperture surface \cite{liu_capa_2025}. This work provides a \gls{CRIS} performance upper bound.} We assume that the direct path between the \glspl{UE} and \gls{BS} is fully obstructed, as commonly assumed in the RIS literature \cite{chrysologou_outage_2022, chrysologou_when_2023, vo_holographic_2024, alaaeldin3}. The \gls{CRIS} is partitioned into $K$ vertical sections, $\mathcal{P}_1$ through $\mathcal{P}_K$. Each partition $\mathcal{P}_k$ spans the full height $H$ and has width $W_k$. The phase shifts of all points $ (x_k, y_k) \in \mathcal{P}_k$ are designed to enhance the channel of $U_k$.

\begin{figure}[!ht]
\vspace{-0.075in}
\centering
\resizebox{\columnwidth}{!}{
\begin{circuitikz}
\tikzstyle{every node}=[font=\LARGE]
\draw [ line width=1pt ] (-11,17.25) rectangle (2.5,12);
\draw [ fill={rgb,255:red,245; green,0; blue,0} , line width=1pt ] (-10.5,16.75) rectangle (-5.75,12.5);
\draw [ fill={rgb,255:red,0; green,194; blue,23} , line width=1pt ] (-2.75,16.75) rectangle (2,12.5);
\node [font=\Huge] at (-4.25,14) {...};
\draw [ line width=1pt](-15.75,13.25) to[short] (-15.75,8.75);
\node[font=\Huge,scale=1.5] at (-15.75,8) {BS};
\draw [ fill={rgb,255:red,0; green,0; blue,0} ] (-15.75,13.5) ellipse (0.25cm and 0.25cm);
\draw [ fill={rgb,255:red,0; green,0; blue,0} ] (-8.25,13.75) ellipse (0.25cm and 0.25cm);
\draw [ fill={rgb,255:red,0; green,0; blue,0} ] (-0.25,13.75) ellipse (0.25cm and 0.25cm);
\node [font=\Huge,scale=1.5] at (-8.25,15.75) {$\mathcal{P}_1$};
\node [font=\Huge,scale=1.25] at (-8.25,14.75) {$(x_1,y_1)$};
\node [font=\Huge,scale=1.5] at (-0.25,15.75) {$\mathcal{P}_K$};
\node [font=\Huge,scale=1.25] at (-0.25,14.75) {$(x_K,y_K)$};
\draw [line width=1pt, <->, >=Stealth] (-10.5,17.75) -- (-5.75,17.75);
\node [scale=1.5,font=\Huge] at (-8.125,18.3) {$W_1$};
\draw [line width=1pt, <->, >=Stealth] (-2.75,17.75) -- (2,17.75);
\node [scale=1.5,font=\Huge] at (-0.375,18.3) {$W_K$};
\draw [line width=1pt, <->, >=Stealth] (3,16.75) -- (3,12.5);
\node [scale=1.5,font=\Huge] at (3.55,14.625) {$H$};
\draw [line width=1pt, ->, >=Stealth] (-8.75,13.75) -- (-15,13.5);
\node[font=\Huge,scale=1.25] at (-12.75,14.25) {$g_1(x_1,y_1)$};
\node[font=\Huge,scale=1.25] at (-12.9,12.25) {$g_K(x_K,y_K)$};
\draw [line width=1pt, ->, >=Stealth] (-0.75,13.75) -- (-15,12.75);
\draw [ fill={rgb,255:red,0; green,0; blue,0} , rounded corners = 3.0] (2.5,11.25) rectangle (4,8.75);
\draw [ fill={rgb,255:red,254; green,255; blue,255} , rounded corners = 3.0] (2.75,11) rectangle (3.75,9);
\node[font=\Huge,scale=1.5] at (3.25,8) {UE $K$};
\draw [ fill={rgb,255:red,0; green,0; blue,0} ] (2.5,11.5) rectangle (2.75,11);
\draw [ fill={rgb,255:red,0; green,0; blue,0} , rounded corners = 3.0] (-4.25,11.25) rectangle (-2.75,8.75);
\draw [ fill={rgb,255:red,254; green,255; blue,255} , rounded corners = 3.0] (-4,11) rectangle (-3,9);
\node[font=\Huge,scale=1.5] at (-3.5,8) {UE 1};
\draw [ fill={rgb,255:red,0; green,0; blue,0} ] (-4.25,11.5) rectangle (-4,11);
\node [font=\Huge] at (-0.25,8.75) {...};
\draw [ fill={rgb,255:red,0; green,0; blue,0} ] (-11,11.5) rectangle (-9.25,8.75);
\draw [ fill={rgb,255:red,0; green,0; blue,0} ] (-9,10.75) rectangle (-7.25,8.75);
\draw [ fill={rgb,255:red,254; green,255; blue,255} ] (-10.75,11.25) rectangle (-10.25,10.75);
\draw [ fill={rgb,255:red,254; green,255; blue,255} ] (-10,11.25) rectangle (-9.5,10.75);
\draw [ fill={rgb,255:red,254; green,255; blue,255} ] (-10.75,10.5) rectangle (-10.25,10);
\draw [ fill={rgb,255:red,254; green,255; blue,255} ] (-10,10.5) rectangle (-9.5,10);
\draw [ fill={rgb,255:red,254; green,255; blue,255} ] (-10.75,9.75) rectangle (-10.25,9.25);
\draw [ fill={rgb,255:red,254; green,255; blue,255} ] (-10,9.75) rectangle (-9.5,9.25);
\draw [ fill={rgb,255:red,254; green,255; blue,255} ] (-8.75,10.5) rectangle (-8.25,10);
\draw [ fill={rgb,255:red,254; green,255; blue,255} ] (-8,10.5) rectangle (-7.5,10);
\draw [ fill={rgb,255:red,254; green,255; blue,255} ] (-8.75,9.75) rectangle (-8.25,9.25);
\draw [ fill={rgb,255:red,254; green,255; blue,255} ] (-8,9.75) rectangle (-7.5,9.25);
\draw [line width=1pt, short] (-15.75,13.5) -- (-14.5,8.75);
\draw [line width=1pt, short] (-15.75,13.5) -- (-17,8.75);
\draw [line width=1pt, short] (-17,8.75) -- (-14.75,9.75);
\draw [line width=1pt, short] (-14.5,8.75) -- (-16.75,9.75);
\draw [short] (-16.25,8) -- (-16.25,8);
\draw [line width=1pt, short] (-16.75,9.75) -- (-15,10.5);
\draw [line width=1pt, short] (-14.75,9.75) -- (-16.5,10.5);
\draw [line width=1pt, short] (-16.5,10.5) -- (-15.25,11.5);
\draw [line width=1pt, short] (-15,10.5) -- (-16.25,11.5);
\draw [line width=1pt, short] (-16.25,11.5) -- (-15.5,12.5);
\draw [line width=1pt, short] (-15.25,11.5) -- (-16,12.5);
\draw [line width=1pt, ->, >=Stealth] (-3.6,11.5) -- (-7.75,13.5);
\draw [line width=1pt, ->, >=Stealth] (3.5,11.5) -- (0.25,13.5);
\draw [line width=1pt, ->, >=Stealth, dashed] (3.25,11.5) -- (-7.5,13.75);
\draw [line width=1pt, ->, >=Stealth, dashed] (-3.4,11.5) -- (-0.25,13.25);
\node[color={rgb,255:red,245; green,0; blue,0},font=\Huge,scale=1.25] at (-6.25,11.35) {$h_{11}(x_1,y_1)$};
\node[color={rgb,255:red,245; green,0; blue,0},font=\Huge,scale=1.25] at (-0.65,10) {$h_{1K}(x_K,y_K)$};
\node[color={rgb,255:red,0; green,194; blue,23},font=\Huge,scale=1.25] at (0.5,11) {$h_{K1}(x_1,y_1)$};
\node[color={rgb,255:red,0; green,194; blue,23},font=\Huge,scale=1.25] at (5.25,12.2) {$h_{KK}(x_K,y_K)$};
\end{circuitikz}}
\vspace{-0.2in}
\caption{A CRIS-enabled UL NOMA system.}
\label{fig:sys_model}
\end{figure}
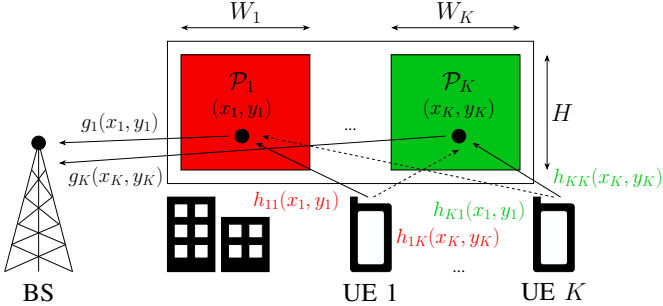

Let $h_{ki}(x_i,y_i)\in\mathbb{C}$ be the normalized channel from $U_k$ to position $(x_i,y_i)\in\mathcal{P}_i$ on the \gls{CRIS}, and $g_i(x_i,y_i)\in\mathbb{C}$ be the normalized channel from $(x_i,y_i)\in\mathcal{P}_i$ to the \gls{BS}. We assume spatially correlated Rayleigh fading so that channels corresponding to any two points $(x_i,y_i)$ and $(x_k,y_k)$ on the \gls{CRIS} are both $\mathcal{CN}(0,1)$ with correlation $\rho(x_i,y_i,x_k,y_k)$. We naturally assume that all \gls{UE}-\gls{CRIS} channels are statistically independent of all \gls{CRIS}-\gls{BS} channels. Let $\Phi_k(x_k,y_k)$ be the \gls{CRIS} reflection coefficient at $(x_k,y_k)\in\mathcal{P}_k$. The phase shifts of $\mathcal{P}_k$ are set to align the channels for $U_k$, such that\footnote{This closed-form expression for the phase shifts of the \gls{CRIS} reduces the complexity by significantly simplifying design and optimization.}
\begin{equation} \label{phase_adjust}
\Phi_k(x_k,y_k) = \mathrm{e}^{-j\left(\angle h_{kk}(x_k,y_k) + \angle g_k(x_k,y_k)\right)}.
\end{equation}
This results in the coherent combining of received signal components for $U_k$ at the destination. The \gls{BS} estimates the cascaded \gls{CRIS} channels of $U_k$ through its designated partition $\mathcal{P}_k$, $h_{kk} g_k$, which it uses to adjust the \gls{CRIS} phase shifts as in (\ref{phase_adjust}) \cite{chrysologou_when_2023, chrysologou_outage_2022}. The angle of the cascaded channel is the sum of the individual channel angles. The signal reflections of $U_k$ via partition $\mathcal{P}_i, \forall i\neq k$ are residual components with zero mean, as the phase shifts of $\mathcal{P}_i$ are not optimized for $U_k$.

The received superimposed \gls{NOMA} signal at the \gls{BS} is thus
\begin{multline} \label{rx_sig}
y = \sum\nolimits_{k=1}^K \sum\nolimits_{i=1}^K \sqrt{\frac{P_k \eta_k}{\nu_k}} \int\nolimits_{x_i=0}^{W_i} \int\nolimits_{y_i=0}^H h_{ki}(x_i,y_i) \\ \times \Phi_i(x_i,y_i)g_i(x_i,y_i)dx_idy_i\,s_k + n,
\end{multline}
where $P_k$ is the transmit power of $U_k$, $s_k$ is the transmitted modulated symbol, $\nu_k$ is a scaling factor that normalizes $s_k$, $n\sim \mathcal{CN}(0, 2\sigma_n^2)$ is the additive white Gaussian noise, $\eta_k= d_{\mathrm{ur},k}^{-\psi} d_{\mathrm{rb}}^{-\psi}$ is the total path loss of the cascaded channel, $d_{\mathrm{ur},k}$ is the distance of the $U_k$-\gls{CRIS} link, $d_{\mathrm{rb}}$ is the distance of the \gls{CRIS}-\gls{BS} link, and $\psi$ is the path-loss exponent.

The modulation symbol, $s_k\in\mathbb{C}$, is drawn from a square \gls{QAM} alphabet, $\mathcal{S}_k$, where the cardinality of $\mathcal{S}_k$ is $M_k$, the modulation order of $U_k$. Both the real and imaginary components of $s_k$ can take values from the set $\{\pm 1, \pm 3, \dots,\pm\sqrt{M_k}-1\}$. The scaling factor $\nu_k=2\left({M_k-1}\right)/3$ is used to normalize $s_k$ to unity.

The superimposed received signal in \eqref{rx_sig} can be rewritten as
\begin{equation}  \label{rx_sig_modified}
y = \sum\nolimits_{k=1}^K \Big( \gamma_{kk} + \sum\nolimits_{i \neq k} \gamma_{ki} \Big) s_k + n,
\end{equation} 
where $\gamma_{kk}$ is the optimized component of the effective channel of $U_k$, such that
\begin{equation}
    \gamma_{kk}\!=\!\sqrt{ \tfrac{P_k \eta_k}{\nu_k} } \!\!\int
    _{\substack{(x_k,y_k) \\ \in\mathcal{P}_k}} \!\left|h_{kk}(x_k,y_k)\right|\!\left|g_{k}(x_k,y_k)\right|dx_kdy_k,\! \label{eq:gammakk}
\end{equation}
while $\gamma_{ki}$ is the random component of the effective channel of $U_k$, such that
\begin{equation}  \label{gamma_ki}
    \!\!\gamma_{ki}\!=\!\!\sqrt{\!\tfrac{P_k \eta_k}{\nu_k} } \!\!\int
    _{\substack{(\!x_i,y_i\!) \\ \in\mathcal{P}_i}}\!\!h_{ki}(\!x_i,y_i\!) \mathrm{e}^{-j\!\angle h_{ii}(\!x_i,y_i\!)}\!\left|g_i(\!x_{i\!},y_{i\!})\right|\!dx_{i}dy_{i}.\!
\end{equation}
The scalar value $\gamma_{ki}$ represents the sum of the cascaded channels from $U_k$ to the \gls{BS} through all points in partition $\mathcal{P}_i, \forall i \neq k$, which can be estimated using a single pilot signal.\footnote{Channel estimation methods for continuous surfaces are fast developing in the literature via approaches, such as reconstruction from discrete samples \cite{fas_samp} and Gaussian process regression \cite{fas_reconst}.} This estimation introduces negligible pilot overhead compared to estimating all individual cascaded channels of $U_k$ through all \gls{CRIS} points in partition $\mathcal{P}_k$, whose pilot overhead scales linearly with the surface area of $P_k$.

\section{Analysis of the RIS effective channels} \label{stat_model}

In this section, we analyze the \gls{CRIS} effective channels and derive terms required to obtain expressions for the \gls{BER}. To facilitate the \gls{BER} derivation, channel alignment is performed to align the phases of the effective channels of all users so that the received superimposed \gls{NOMA} signal has a \gls{QAM}-like constellation shape. This alignment is achieved by applying a phase shift of $\beta_i$ to all points in partition $\mathcal{P}_i$, such that they cancel the imaginary components of all effective channels as
\begin{equation}  \label{align}
\Im \Big( \sum\nolimits_{i=1}^K \gamma_{ki} e^{j \beta_i} \Big) = 0, \quad \! \forall k \in \{1, \dots, K\}. 
\end{equation}
To solve (\ref{align}), the system of equations is transformed to a least-squares unconstrained minimization problem, as
\begin{equation} \label{least_sq}
\min_{\beta_1, \dots, \beta_K} \quad \sum\nolimits_{k=1}^K \Big( \Im \Big( \sum\nolimits_{i=1}^K \gamma_{ki} e^{j \beta_i} \Big) \Big)^2.
\end{equation}
Problem (\ref{least_sq}) can be solved by a standard nonlinear optimization solver, such as the trust-region or Levenberg–Marquardt algorithms. The solution always converges to a stationary minimal point where the value of the objective function is zero, as $\gamma_{kk} \gg \gamma_{ki}$, so $\gamma_{kk}$ dominates the random components for small $\beta_i$. Thus, the aligned effective channel of $U_k$ is
\begin{equation} \label{real_approx}
h_k^{\mathrm{eff}} = \sum\nolimits_{i=1}^K \gamma_{ki} e^{j \beta_i} \approx  \gamma_{kk} + \sum\nolimits_{ i=1, i \neq k }^K \Re(\gamma_{ki}), \quad \forall k,
\end{equation}
and the received signal becomes $y = \sum\nolimits_{k=1}^K h_k^{\mathrm{eff}} x_k + n$. As $\gamma_{kk}$ and $\gamma_{ki}$ are effectively uncorrelated\footnote{As $\gamma_{kk}$ and $\gamma_{ki}$ are obtained by integrating over disjoint partitions, they become asymptotically uncorrelated. As $W_k$ grows, inter-partition distances increase and the 
remaining cross-correlations average out to negligible values.}, the \gls{CF} of $h^{\mathrm{eff}}_k$ is
\begin{equation} \label{cf_eff_ch}
\Phi_{h^{\mathrm{eff}}_{k}}(z) = \Phi_{\gamma_{kk}}(z) \prod\nolimits_{i\neq k} \Phi_{\Re(\gamma_{ki})}(z).  
\end{equation}
The exact distributions of $\gamma_{kk}$ and $\Re(\gamma_{ki})$ are intractable. We instead use the highly accurate gamma approximation for $\gamma_{kk}$ and Gaussian approximation for $\Re(\gamma_{ki})$ \cite{alaaeldin3,inwood_continuous_2025}, giving 
\begin{equation}
    \Phi_{\gamma_{kk}}(z) \approx \tfrac{1}{(1 - j \theta_k z)^{\alpha_k}}, \label{eq:CFgammakk}
\end{equation}
\begin{equation}
    \Phi_{\Re(\gamma_{ki})}(z)\approx \mathrm{exp}\!\Big(\!j\Re(\mathbb{E}[\gamma_{ki}])z+\tfrac{\mathrm{Var}[\Re(\gamma_{ki})]z^2}{2}\!\Big)\!,  \label{eq:CFgammaki} 
\end{equation}
where $\theta_k = \frac{\mathrm{Var}[\gamma_{kk}]}{\mathbb{E}[\gamma_{kk}]}$ is the gamma scale parameter and $\alpha_k = \theta_k \mathbb{E}[\gamma_{kk}]$ is the gamma shape parameter. Therefore, the first and second moments of $\gamma_{kk}$ and $\Re(\gamma_{ki})$ are required.

\subsection{First moments}
Adapting \cite[Eq. (23)]{inwood_continuous_2025}, the first moment of $\gamma_{kk}$ is
\begin{align}
    \mathbb{E}[\gamma_{kk}] &{=} \sqrt{\!\tfrac{P_k\eta_k}{\nu_k}}\!\!\!\int
    _{\substack{(x_i,y_i) \\ \in\mathcal{P}_i}}\!\!\mathbb{E}[|h_{kk}(x_k,y_k)|]\mathbb{E}[|g_k(x_k,y_k)|]dx_kdy_k,\notag \\
    &{=}\sqrt{\!\tfrac{P_k\eta_k}{\nu_k}}\frac{\pi}{4}W_kH. \label{eq:Egammakk}
\end{align}
Since $\mathbb{E}[h_{ki}(x_k,y_k)]=0$, $\mathbb{E}[\gamma_{ki}] =0$.

\subsection{Second moments}
The second moment of $\gamma_{kk}$ is
\begin{equation*}
    \mathrm{Var}[\gamma_{kk}] = \mathbb{E}[\gamma_{kk}^2]-\mathbb{E}[\gamma_{kk}]^2. 
\end{equation*}
As $\Re(\gamma_{ki})$ is zero-mean, and as both the real and imaginary components of $\gamma_{ki}$ are identically distributed, we have 
\begin{equation*}
     \mathrm{Var}[\Re(\gamma_{ki})] = \tfrac{1}{2}\mathbb{E}[|\gamma_{ki}|^2].
\end{equation*}
The expression for $\mathbb{E}[\gamma_{kk}]$ was stated in \eqref{eq:Egammakk}, and the expressions for $\mathbb{E}[|\gamma_{kk}|^2]$ and $\mathbb{E}[|\gamma_{ki}|^2]$ are given below. 

\begin{proposition} \label{prop}
   Assuming isotropic correlation, $\mathbb{E}\left[\gamma_{kk}^2\right] = \Omega(k)$  and $\mathbb{E}\left[|\gamma_{ki}|^2\right] = \Omega(i)$,
    with $\Omega(x)$ given in \eqref{eq:Omegax} and 
\begin{figure*}[t]
\footnotesize
\begin{multline}
\label{eq:Omegax}
     \Omega(x)\!=\!\frac{P_x\eta_x\pi^2}{4\nu_x}\!\bigg(\!\int_0^H\!\!\!r_xg_x(r_x)\!\left[W_xH\frac{\pi}{2}\!-\!(W_x\!+\!H)r_x\!+\!\frac{r_x^2}{2}\right]\!dr_x+\!\int_H^{W_x}\!\!\!r_xg_x(r_x) \!\bigg[W_xH\sin^{-1}\!\left(\!\frac{H}{r_x}\!\right) \!+W_x\sqrt{r_x^2\!-\!H^2}-Wr_x-\frac{H^2}{2}\bigg]dr_x\\-\int_H^{\sqrt{W_x^2+H^2}}\!\!\!r_xg_x(r_x)\!\left[W_xH\!\left(\!\cos^{-1}\!\left(\!\frac{W_x}{r_x}\!\!\right)\!-\!\sin^{-1}\!\!\left(\!\frac{H}{r_x}\!\right) \!\right)\!+\!\frac{W_x^2\!+\!H^2\!+\!r_x^2}{2}-W_x\sqrt{r_x^2\!-\!H^2}-H\sqrt{r_x^2\!+\!W_x^2}\right]\!dr_x\!\bigg)\!,
\end{multline}
\normalsize
\hrulefill
\vspace{-0.2 in}
\end{figure*}
    \begin{equation}
        g_k(r_k) = {}_2F_1^{\,2}\!\left(-0.5,-0.5,1,|\rho(r_k)|^2\right), \notag
    \end{equation}
    \begin{equation}
        \!g_i(r_{\!i})\!=\!\rho^{2\!}(r_{\!i}){}_2F_1\!\!\left(\!-0.5,\!-0.5,\!1,\!|\rho(r_{\!i})|^2\right)\!{}_2F_1\!\!\left(\!0.5,\!0.5,\!2,\!|\rho(r_{\!i})|^2\right)\!\!,\! \notag
    \end{equation}
    with $r_k{=}\sqrt{\!(\!x_k\!{-}x_k')^2{+}(y_k\!{-}y_k')^2\!}$ and $r_{\!i}{=}\sqrt{\!(\!x_i{-}x_i')^2{+}(y_i{-}y_i')^2\!}$.
\end{proposition}

\renewcommand\qedsymbol{$\blacksquare$}

\begin{proof}
Substituting $\gamma_{kk}$ with (\ref{eq:gammakk}), $\mathbb{E}[\gamma_{kk}^2]$ can be expressed as
\begin{multline}
    \mathbb{E}[\gamma_{kk}^2] = \tfrac{P_k\eta_k}{\nu_k}\!\!\int\nolimits
    _{\substack{(x_k,y_k) \\ \in\mathcal{P}_k}} \int\nolimits
    _{\substack{(x_k',y_k') \\ \in\mathcal{P}_k}}\!\mathbb{E}[|h_{kk}(x_k,\!y_k)||h_{kk}(x_k',\!y_k')|] \\ \times\mathbb{E}[|g_k(x_k,y_k)||g_k(x_k',y_k')|]dx_kdy_kdx_k'dy_k'. \notag
\end{multline}
Both expectations involve products of identically distributed correlated Rayleigh random variables with the same correlation function. Therefore, applying \cite[Eq. (4.19)]{miller_complex_1974} yields
\begin{equation}
    \mathbb{E}[\gamma_{kk}^2]{=}\tfrac{\!P_k\eta_k\pi^2}{16\nu_k}\!\!\!\int\nolimits
    _{\substack{(x_k,y_k) \\ \in\mathcal{P}_k}}\! \int\nolimits
    _{\substack{(x_k',y_k') \\ \in\mathcal{P}_k}}\!g_{k}(\!x_{\!k},\!y_k,\!x_{k}',\!y_k'\!)dx_{\!k}dy_{k}dx_{k}'dy_k'.\!\! \label{eq:Egammakk2_appen1}
\end{equation}
with
\begin{equation*}
    g_k(x_k,y_k,x_k',y_k')={}_2F_1^2(-0.5,-0.5,1,|\rho(x_k,y_k,x_k',y_k')|^2).
\end{equation*}
Substituting $\gamma_{ki}$ with (\ref{gamma_ki}), $\mathbb{E}[|\gamma_{ki}|^2]$ can be expressed as
\begin{multline*}
\mathbb{E}[|\gamma_{ki}|^2] = \tfrac{P_k\eta_k}{\nu_k} \! \int\nolimits_{\substack{(x_i,y_i) \\ \in \mathcal{P}_i}} \! \int\nolimits_{\substack{(x_i',y_i') \\ \in \mathcal{P}_i}} \! \mathbb{E} \! \left[ \mathrm{e}^{j\left(\!\angle\!h_{ii}\!(\!x_{i\!}',y_{i\!}'){-}\!\angle\!h_{ii}\!(\!x_i,y_i\!)\!\right)}\!\right]  \\
\times\mathbb{E}[h_{ki}(\!x_{i},\!y_{i\!})h^*_{ki}(\!x_{i},\!y_{i\!})]\mathbb{E}[|g_{i\!}(\!x_{i},\!y_{i\!})||g_{i\!}(\!x_i',\!y_i')|]dx_idy_idx_i'dy_i'.\!
\end{multline*}
As it is the product of two correlated random variables, $\mathbb{E}[h_{ki}(x_{i},y_{i})h^*_{ki}(x_{i},y_{i})] {=} \rho(x_i,y_i,x_i',y_i')$. For $\mathbb{E}[\gamma_{kk}^2]$, $\mathbb{E}[|g_{i}(x_{i},y_{i})||g_{i}(x_i',y_i')|]$ can be found the same way as $\mathbb{E}[|g_k(x_k,y_k)||g_k(x_k',y_k')|]$, whereas \cite[Eq.~(4.26)] {miller_complex_1974} gives $\mathbb{E}\big[\mathrm{e}^{j(\angle h_{ii}(x_{i}',y_{i}'){-}\angle h_{ii}(x_i,y_i))}\big]$. Therefore,
\begin{equation}
    \mathbb{E}[|\gamma_{ki}|^2]{=}\tfrac{\!P_k\eta_k\pi^2}{16\nu_k}\!\!\int\nolimits
    _{\substack{(x_i,y_i) \\ \in\mathcal{P}_i}}\!\int\nolimits
    _{\substack{(x_i',y_i') \\ \in\mathcal{P}_i}}\!g_{i}(\!x_{\!i},\!y_i,\!x_{\!i}',\!y_i')dx_{\!i}dy_{i}dx_{\!i}'dy_i', \label{eq:absgammaki}
\end{equation} 
where
\begin{multline}
g_i(x_{i},y_{i},x_{i}',y_{i}')={}_2F_1\!\left(0.5, 0.5, 2, |\rho(x_{i}, y_{i},x_{i}', y_{i}')|^2\right) \\ \times {}_2F_1\!\left(-0.5, -0.5, 1, |\rho(x_i,y_i,x_i',y_i')|^2\right)\!\rho^2(x_{i},y_{i},x_{i}',y_{i}').\notag 
\end{multline}
Both \eqref{eq:Egammakk2_appen1} and \eqref{eq:absgammaki} are integrals of the same general form over different subsurfaces. This type of integral is solved in \cite[Appendix C]{inwood_continuous_2025}. As isotropic correlation is assumed, the quadruple integrals can be transformed into a sum of single integrals over the distance between pairs of points, $r$, resulting in  $\Omega(x)$ in \eqref{eq:Omegax}.  Substituting $k$ and $i$ as the argument $x$ gives the final result for $\mathbb{E}[\gamma_{kk}^2]$ and $\mathbb{E}[|\gamma_{ki}|^2]$, respectively.
\end{proof}

\section{PA and RIS area splitting-based optimization scheme for CRIS-NOMA}  \label{analys_opt}

In this section, we use the expressions derived in Sec. \ref{stat_model} to obtain analytical expressions for the \gls{BER} of each \gls{UE}. We use these \gls{BER} expressions to optimize the transmit powers, $P_k$, and the \gls{CRIS} partition widths, $W_k$, to minimize the average \gls{BER} of all \glspl{UE} at the \gls{BS}.\footnote{The extension to the multi-antenna BS case follows the same design and BER analysis as the single-antenna scenario. The key difference is the introduction of a BS receive combiner, which can be jointly optimized with the \gls{CRIS} phase shifts via alternating maximization. The resulting, empirically characterized, effective channel statistics can then be used in the same BER-based power allocation procedure as in the single-antenna case.} From \cite{alaaeldin3}, the unconditional average \gls{BER} of $U_k$ for channel-aligned \gls{UL} \gls{NOMA} is
\begin{equation}  \label{ber_express}
\mathrm{BER}_{U_k} {=}\! \sum\nolimits_{q=1}^{N_k} \!\frac{c_{q}}{2} {+} \frac{c_{q}}{ \pi} \!\int_0^{\infty} \!\!\!\Re \!\bigg(\!\! \frac{j \mathrm{e}^{ -z^2/2 } }{z}  \Phi_{X_q}\!\bigg( \frac{z}{\sigma_n} \bigg) \!\!\bigg) \! dz,
\end{equation}
where $N_k$ is the total number of $Q(\cdot)$ functions in \eqref{ber_express}, $a_{kq}$ and $c_{q}$ are constants that depend on $K$ and $M_k$, and $\Phi_{X_q} (z) = \prod\nolimits_{k = 1}^K \Phi_{h^{\mathrm{eff}}_k} (a_{kq} z)$.
The proposed optimization strategies aim to eliminate the \gls{BER} floors inherent in  \gls{SIC}-based \gls{UL} \gls{PD}-\gls{NOMA} systems. Using (\ref{ber_express}), we derive optimized values of $P_k$ and $W_k$ for each \gls{UE} that minimize the average (sum) BER\footnote{Note that we minimize the average (sum) of the \glspl{BER} of the \glspl{UE} to ensure that no \gls{UE}'s \gls{BER} remains high, as a single large \gls{BER} would dominate the sum. This formulation is simpler than minimizing the maximum \gls{BER}, however, we compare against the min-max solution in the results in Sec. \ref{Sim}.} of all users while satisfying individual \gls{UL} transmit power constraints. The optimization is more stable and accurate when expressed in the log–log domain, where both $P_k$ and the cost function are expressed in dB. This transformation leads to a smoother objective function and faster convergence with gradient descent-based methods. The UL PA problem is thus
\begin{subequations} \label{opt_prob}
\begin{align}
&\!\!\!\min_{P_k, W_k}\!\! 10 \log_{10}\!\! \bigg( \sum_{k=1}^K \mathrm{BER}_{U_k}\! \Big(\!\! 10^{\frac{P_1}{10}}\!,\! \cdots\!,\! 10^{\frac{P_K}{10}}\!,\! W_1,\! \cdots\!,\! W_{\!K}\! \!\Big)\!\!\! \bigg)\!,\!\! \label{cost_fun} \\ 
&\!{\mathrm {s.t.}} \,\, P_k  \leq  P_{\mathrm{dBm}}^{\mathrm{max}}, \qquad \forall k, \\
&\quad\,\,\, 0 \leq W_k  \leq  W, \qquad \forall k, \\
& \quad\,\,\sum\nolimits_{k=1}^K W_k = W,
\end{align}
\end{subequations}
where $P_{\mathrm{dBm}}^{\mathrm{max}}$ is the maximum available \gls{UL} transmit power in dBm and $\mathrm{BER}_{U_k}$ is given in (\ref{ber_express}). The optimization parameters $P_k$ and $W_k$ in (\ref{cost_fun}) come from $\Phi_{X_q}(\cdot)$ in (\ref{ber_express}), which is a function of $\Phi_{h^{\mathrm{eff}}_k}(\cdot)$ which in turn is a function of $\Phi_{\gamma_{kk}}(\cdot)$ and $\Phi_{\Re(\gamma_{ki})}(\cdot)$ in (\ref{cf_eff_ch}), and both are functions of $P_k$ and $W_k$.

To handle the constraints, (\ref{opt_prob}) is transformed into an unconstrained optimization problem by constructing the Lagrangian,
\begin{multline}  \label{largrange}
L(\mathbf{p}, \mathbf{w}, \boldsymbol{\xi}, \boldsymbol{\delta}, \omega) = f(\mathbf{p}) + \sum\nolimits_{k=1}^K \xi_k \max(0, P_k {-} P_{dBm}^{\mathrm{max}}) \\
+\!\!\sum\nolimits_{k=1}^K \!\delta_k\! \max(0, W_k {-} W)\!+\!\omega \Big(\! \sum\nolimits_{k=1}^K\!\! W_k\!-\!W\!\Big)^2\!, \!\!\! 
\end{multline}
where $\boldsymbol{\xi}=[\xi_1,\dots,\xi_K]^T$, $\boldsymbol{\delta}=[\delta_1,\dots,\delta_K]^T$ and $\omega$ are the Lagrange multipliers, $f(\cdot)$ is the cost function in (\ref{opt_prob}), $\mathbf{p} = [P_1, \dots, P_K]^T$ and $\mathbf{w} = [W_1, \dots, W_K]^T$. The optimization problem is solved iteratively by setting low initial values for the penalty weights $\boldsymbol{\xi}$, $\boldsymbol{\delta}$ and $\omega$. Next, the Lagrangian in (\ref{largrange}) is minimized in each iteration using a gradient descent-based method as it is continuous and has continuous first derivatives. The penalty weights are increased in each iteration if the corresponding constraints are violated. After multiple iterations, the algorithm converges to the minimum weights at which \eqref{cost_fun} is minimized while the constraints are satisfied.

The computational complexity of problem (\ref{opt_prob}) is dominated by the cost of the gradient evaluation of the Lagrangian in (\ref{largrange}). The Lagrangian gradient is approximated via forward finite differences, requiring $\mathcal{O}(K)$ additional function evaluations per iteration. The cost of one Lagrangian evaluation $C_L$ is dominated by the computation of $\sum_{k=1}^K \mathrm{BER}_{U_k}$. With $I$ inner and $J$ outer iterations, the computational complexity of the proposed algorithm is $\mathcal{O}(J I K C_L)$.

The complexity of computing $C_L$ is dominated by calculating the \gls{BER} expression in (\ref{ber_express}). To reduce the cost, the Gaussian hypergeometric function used to calculate the second moments of $\gamma_{kk}$ and $\gamma_{ki}$ is approximated using the polynomial series expansion in \cite[Eq. (15.2.1)]{nist_digital_2026} truncated to $L$ terms. This expansion is very accurate when $|z| < 1$, which is true here as $z$ is the squared isotropic correlation. Thus, a value such as $L = 6$ is sufficient for a highly accurate approximation.


\section{Simulation Results} \label{Sim}

Figures \ref{2-users} and \ref{3-users} plot the average \gls{BER} of each \gls{UE} against the transmit power for a 2 and 3 \gls{UE} system, respectively. The noise variance is $\sigma_n^2 {=} 10^{-9}$ mW \cite{Qingqing2}, the path loss exponent is $\psi {=} 2.2$, $d_\mathrm{rb} {=} 30$ m, $d_{\mathrm{ur},1} {=} 20$ m and $d_{\mathrm{ur},3} {=} 220$ m. In Fig. \ref{2-users}, $d_{\mathrm{ur},2} {=} 50$ m and in Fig. \ref{3-users}, $d_{\mathrm{ur},2} {=} 70$ m. Due to its close proximity to the \gls{CRIS}, $U_1$ employs 64-QAM to increase throughput, while $U_2$ and $U_3$ employ 16-QAM. The \gls{CRIS} height is fixed at $H {=} 5\lambda$, where $\lambda$ is the signal wavelength, while the width, $W$, varies across each figure. The carrier frequency used in the simulations is $f {=} 28$ GHz. The sinc correlation model is used, such that $\rho(x_i,y_i,x_k,y_k) = \mathrm{sinc}(\tfrac{2}{\lambda} \sqrt{(x_i - x_k)^2 + (y_i - y_k)^2})$. We compare the performance of the \gls{CRIS} with a \gls{DRIS} with an element spacing of $\lambda/2$, and having the same surface area for a fair comparison.

Each sub-figure compares four methods for a single \gls{UE}: the joint optimization approach in \eqref{opt_prob} (JO), optimization only over \gls{CRIS} area partitions (AO), no optimization (NO), and a time-slotted \gls{OMA} baseline. In scenarios where the \gls{CRIS} partition size or transmit power is not optimized, the resource is shared equally among \glspl{UE}. The four methods are applied and compared under different \glspl{CRIS} sizes in each figure. Figure \ref{2-users} also shows results for the minimization of the maximum \gls{BER} (MM), as a comparison to minimizing the average \gls{BER} in (\ref{cost_fun}).

For a fair comparison, the modulation order of $U_k$ in the \gls{OMA} system is chosen as
$M_k^{\textnormal{OMA}} = \left(M_k^{\textnormal{NOMA}}\right)^K$, such that the number of bits per symbol for \gls{OMA} is that of \gls{NOMA} multiplied by $K$. This ensures that both systems achieve the same overall bit rate. Moreover, the transmit power of $U_k$ in the \gls{OMA} system is scaled by a factor of $K$ so that the average transmit power is identical in \gls{OMA} and \gls{NOMA}. To maintain fairness in terms of channel estimation overhead, in the \gls{OMA} scenario, we assume that the \gls{BS} only has knowledge of $\mathcal{P}_k$ for $U_k$, as in the \gls{NOMA} case, rather than of the entire \gls{CRIS}.

\begin{figure}[t]
\centering
\includegraphics[width=0.9\columnwidth]{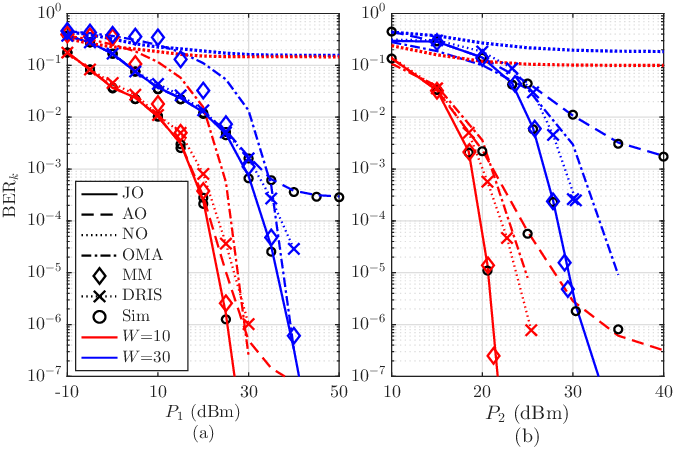}
\vspace{-1em}
\caption{\glspl{BER} for two-user systems: (a) $U_1$, (b) $U_2$.}
\label{2-users}
\vspace{-1em}
\end{figure}

\begin{figure*}
\centering
\includegraphics[width=1.6\columnwidth]{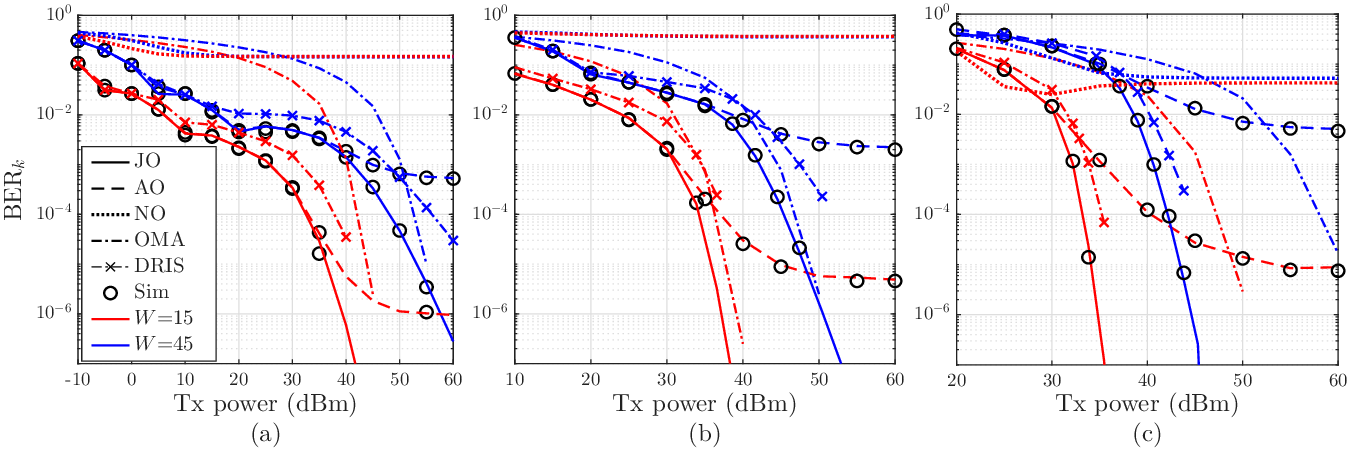}
\vspace{-0.5em}
\caption{\glspl{BER} for three-user systems: (a) $U_1$, (b) $U_2$, (c) $U_3$.}
\vspace{-0.2 in}
\label{3-users}
\end{figure*}

Figures \ref{2-users} and \ref{3-users} demonstrate excellent agreement between the analytical results using \eqref{ber_express} and simulations. This confirms the accuracy of our approximation of $\Phi_{h^{\mathrm{eff}}_{k}}(z)$, which enables the tractable \gls{BER} expression required for optimization.

Without optimization, \gls{NOMA} exhibits persistent error floors due to employing \gls{SIC}. Residual errors from \gls{SIC} persist even as the transmit power increases and noise impact diminishes, resulting in the observed error floor. This floor exceeds 0.1 for the strongest \gls{UE}, meaning over 1 in 10 bits are incorrectly detected, highlighting the need for optimization for reliable operation. Equal \gls{PA} overlooks variations in channel quality among \glspl{UE}, reducing the effectiveness of \gls{SIC}. Increasing the \gls{CRIS} size without proper optimization yields limited \gls{BER} improvement, as \gls{IUI} remains significant. Optimizing the \gls{CRIS} partitions while maintaining equal transmit \gls{PA} among \glspl{UE} leads to significant improvements in \gls{BER} performance. This gain is especially evident for larger \gls{CRIS} sizes and for stronger \glspl{UE} as they are assigned larger sections of the \gls{CRIS} since they are first decoded in the \gls{SIC} order. However, the error floor is not completely eliminated. This occurs because, at high \gls{SNR} levels, the optimization algorithm allocates larger \gls{CRIS} partitions to the stronger \gls{UE} and smaller \gls{CRIS} partitions to the weaker \glspl{UE} to suppress \gls{IUI} and mitigate the error floor. However, the smaller partition sizes reduce channel hardening for the weaker \glspl{UE}, and this increased variability degrades its \gls{BER} performance.  The error floor emerges when further improvements in the \gls{BER} of the stronger \glspl{UE} come at the expense of equivalent \gls{BER} degradation of the weaker \glspl{UE}, limiting the overall performance.

The joint optimization of the \gls{CRIS} partition and \gls{PA} completely removes the error floor for \glspl{BER} of at least $10^{-7}$. Optimizing \gls{PA} complements \gls{CRIS} partitioning by enabling interference reduction without sacrificing channel hardening. PA manages \gls{IUI} independently in the \gls{PD}, enabling interference reduction without reducing \gls{CRIS} resources. This approach preserves or even enhances channel hardening while simultaneously improving \gls{BER} performance. Figure \ref{2-users} also shows that minimizing the maximum \gls{BER} and minimizing the average \gls{BER} in (\ref{cost_fun}) converge to the same result at high \gls{SNR}, but (\ref{cost_fun}) provides better performance in the low \gls{SNR} region. The superiority of the \gls{CRIS} over a \gls{DRIS} is also clear.

Optimized \gls{NOMA} outperforms \gls{OMA} in all scenarios studied by enabling multiple \glspl{UE} to simultaneously access the full resource block without incurring excessive \gls{IUI}. The extra resource per user allows \gls{NOMA} to use lower-order modulation schemes, reducing the likelihood of symbol errors. The JO of \gls{CRIS} partitioning and \gls{PA} is very effective at managing \gls{IUI}, allowing the benefits of increased resource sharing to outweigh interference and provide superior \gls{BER} performance.

\section{Conclusions}  \label{Conc}
This paper investigated a \gls{CRIS}-enabled \gls{UL} \gls{NOMA} system under spatially correlated fading. We analyzed the effective channels and derived \gls{BER} expressions using a \gls{CF} approximation. A joint \gls{PA} and \gls{CRIS} partitioning optimization scheme was developed to minimize the average \gls{BER} and mitigate the BER floors caused by \gls{IUI}. Simulation results confirmed the accuracy of the analysis and the superiority of the optimized \gls{CRIS}-\gls{NOMA} scheme over both a traditional \gls{OMA} system, and a \gls{NOMA} system involving a \gls{DRIS}, highlighting the potential of combining \gls{CRIS} and \gls{NOMA} in future networks.

\end{document}